\documentclass[11pt]{article}
\usepackage[a4paper]{geometry}
\usepackage{amssymb}
\usepackage{amsmath}
\usepackage{amsthm}
\usepackage{nicefrac}
\usepackage{url}

\newtheorem{definition}{Definition}
\newtheorem{lemma}{Lemma}
\newtheorem{theorem}{Theorem}
\newtheorem{proposition}{Proposition}
\newtheorem{corollary}{Corollary}

\def\is{\textsc{max independent set}}
\def\vc{\textsc{min vertex cover}}
\def\ds{\textsc{min dominating set}}

\def\sp{\textsc{max set packing}}
\def\cb{\textsc{max complete bipartite subgraph}}
\def\msat{\textsc{min sat}}
\def\fvs{\textsc{min feedback vertex set}}
\def\ids{\textsc{min independent dominating set}}
\def\lcol{\textsc{max $\ell$-colorable induced subgraph}}
\def\mus{\textsc{max unused sets}}
\def\m3sat{\textsc{max 3-sat}}
\def\mmvc{\textsc{max minimal vertex cover}}

\def\ps{\textsc{max induced planar subgraph}}

\def\sol{\mathrm{sol}}
\def\opt{\mathrm{opt}}
\def\argmax{\mathrm{argmax}}

\let\leq\leqslant

\title{\textbf{Sparsification and subexponential approximation}}

\author{\'Edouard Bonnet \hspace*{1cm} Vangelis~Th.~Paschos\footnote{Also, Institut Universitaire de France} \\
PSL Research University, Universit\'e Paris-Dauphine \\
LAMSADE, CNRS UMR 7243, France \\
\texttt{edouard.bonnet@dauphine.fr,paschos@lamsade.dauphine.fr}}

\begin{document}

\maketitle

\begin{abstract}

Instance sparsification is well-known in the world of exact computation since it is very closely linked to the \textit{Exponential Time Hypothesis}. In this paper, we extend the concept of sparsification in order to capture subexponential time approximation. We develop a new tool for inapproximability, called approximation preserving sparsification and use it in order to get strong inapproximability results in subexponential time for several fundamental optimization problems as \is{}, \ds{}, \fvs{}, and \textsc{min set cover}.

\end{abstract}

\section{Introduction}\label{intro}

The most common way to cope with intractability in complexity theory is the design and analysis of efficient approximation algorithms. 
The main stake of such algorithms is to ``fastly'' compute feasible solutions for the hard problems tackled (avoiding so, if possible, long and time-consuming computations needed for determining optimal solutions). 
The values of these solutions must be as ``close'' as possible to the optimal values. 

Historically, the first research program dealing with approximation, was the \textit{polynomial time approximation theory} founded back in~1974 with the seminal paper~\cite{jo}. 
%The approximation ratio of an algorithm solving an optimization problem~$\Pi$ is the ratio ``value of the approximate solution computed over the optimal value''. 
%This program has mobilized numerous researchers in theoretical computer science and has motivated very intensive and fruitful research that greatly contributed to enhancing and deepening our understanding of the nature of intractable problems. 
Since the early~90's, using the celebrated PCP theorem~(\cite{motwaj}), numerous natural hard optimization problems have been proved to admit more or less pessimistic inapproximability results. 
For instance, for any $\epsilon>0$, \is{} is inapproximable within approximation ratio~$n^{\epsilon-1}$, unless $\textbf{P} = \textbf{NP}$~(\cite{zucker}). 
Similar results, known as \textit{inapproximability} or \textit{negative results}, have been provided for numerous other paradigmatic optimization problems. %, such as \textsc{min coloring}~\cite{zucker}. 
%Such inapproximability results exhibit large gaps between what it is possible to do in polynomial time and what would probably become possible if one allows super-polynomial time.

To remedy to this pessimistic context, two complementary research programs,  dealing with super-polynomial approximation, came to be added in the approximation landscape. 
%
%\begin{itemize}
%\item 
The first one, called \textit{parameterized approximation}, handles approximation by fixed parameter algorithms. 
%We say that a minimization (maximization, respectively) problem~$\Pi$, together with a parameter~$k$, is {\em parameterized $r$-approximable} if there exists an FPT-time algorithm which computes a solution of size at most (at least, respectively) $rk$ whenever the input instance has a solution of size at most (at least, respectively) $k$. 
This line of research was initiated by three independent works~\cite{dofemcciwpec,caihuiwpec,chgrogruiwpec}. 
%For an excellent overview on early stages of the topic, see~\cite{marx-approx}. 
%Since then, very important research has been conducted on several aspects (both computational and structural) of parameterized approximation (see, for example,~\cite{effapprox,subexpo-fpt-inapprox-ipec13,DBLP:journals/corr/BonnetEPT13,DBLP:journals/corr/ChitnisHK13,Downey:appDS,safeapprox,DBLP:journals/corr/HajiaghayiKK13}). 
%
%\item 
The second research program, called \textit{moderately exponential approximation}, 
%the core question is whether a problem is approximable in moderately exponential time while such approximation is impossible in polynomial time. Suppose a problem is solvable in time~$O^*(\gamma^n)$ (notation~$O^*(\cdot)$ ignores polynomial factors), but it is \textbf{NP}-hard to approximate within ratio~$r$. Then, one 
seeks, given a problem~$\Pi$, for $r$-approximation algorithms with running time significantly faster than those of exact algorithms computing optimal solution for~$\Pi$. This issue has been independently developed by~\cite{effapproxcah,effapprox,CyganP10,FurerGK09}.
%\end{itemize}

However, a fundamental question remained globally unanswered by both of them. Is \textit{subexponential} approximation possible for some paradigmatic optimization problems as, for instance, \is{}, \vc{}, or \ds{}? A first answer about \is{} and \vc{} has been provided in~\cite{subexpo-fpt-inapprox-ipec13} where it is proved the following.
\begin{theorem}\label{negipec13thm} \cite{subexpo-fpt-inapprox-ipec13}
Under ETH\footnote{The \textit{Exponential Time Hypothesis}~({ETH})~\cite{impa} postulates that there exists an $\epsilon > 0$ such that no algorithm solves \textsc{3-sat} in time~$O^*(2^{\epsilon n})$, where~$n$ is the number of variables. This is a widely-acknowledged computational assumption.}, in graphs of order~$n$:
\begin{enumerate}
\item\label{isneg} for any positive constant~$r$ and any $\delta>0$, there is no $r$-ap\-p\-r\-o\-x\-i\-ma\-ti\-on algorithm for \is{} running in time~$O^*(2^{n^{1-\delta}})$;
%\item for any $r>1$ and any $\delta>0$, there is no $r$-approximation algorithm for \col{} running in time~$O^*(2^{n^{1-\delta}})$;
\item\label{vcneg} for any $\epsilon>0$ and any $\delta>0$, there is no $(\nicefrac{7}{6}-\epsilon)$-ap\-p\-ro\-x\-i\-ma\-ti\-on algorithm for \vc{} running in time~$O^*(2^{n^{1-\delta}})$.
\end{enumerate}
\end{theorem}
The result of Item~\ref{isneg} of Theorem~\ref{negipec13thm} has been powerfully improved by~\cite{2013arXiv1308.2617C}, where a very clever implementation of~PCP~\cite{DBLP:conf/focs/MoshkovitzR08} leads to the following theorem. 
%(by the definition given above, approximation ratios for maximization problems are smaller than~1). 
\begin{theorem}\label{Chalermsook13thm} \cite{2013arXiv1308.2617C}
Under ETH, in graphs of order~$n$ with maximum degree~$\Delta$:
\begin{enumerate} 
\item\label{Chalermsook13thm1} (General graphs) for any $\delta>0$ and any~$r$ larger than some constant, any $r$-ap\-p\-ro\-x\-i\-ma\-ti\-on algorithm for \is{} runs in time at least $O^*(2^{\nicefrac{n^{1-\delta}}{r^{1+\delta}}})$;
\item\label{chalermsook} ($\Delta$-sparse graphs\footnote{Graphs where the maximum degree is bounded by~$\Delta$.}) for any sufficiently small $\varepsilon > 0$, there exists a constant~$\Delta_\varepsilon$, such that for any $\Delta \geqslant \Delta_\varepsilon$, \is{} on $\Delta$-sparse graphs is not $\Delta^{1-\varepsilon}$-approximable in time $O^*(2^{\nicefrac{n^{1-\varepsilon}}{\Delta^{1+\varepsilon}}})$. 
\end{enumerate}
\end{theorem}
%Note that the result of Item~\ref{Chalermsook13thm1} of Theorem~\ref{Chalermsook13thm} nearly matches the upper bound of~$2^{\nicefrac{n}{r}}$. Indeed, as it is proved in~\cite{effapproxcah,kowalik}, \textit{any maximization hereditary problem (\is{} is such a problem) can be approximately solved within ratio~$r$ in time~$O^*(2^{\nicefrac{n}{r}})$}.
%Theorem~\ref{Chalermsook13thm} proves, in fact, that for any $r \geqslant n^{-1/2-\delta}$ and for any $\delta >0$, \is{} is not $r$-approximable in~$O^*(2^{n^{1-\delta}/r^{1+\delta}})$, unless ETH fails. 
%
Our goal in this paper is to introduce a new technique based upon the development of a novel notion of \textit{approximation preserving sparsification} that extends the scope of the classical sparsifiation of~\cite{impa}. 
%A preliminary though incomplete version of this idea has been presented in~\cite{subexpo-fpt-inapprox-ipec13}, in order to build a class of problems that are equivalent regarding subexponential approximability. Here, we improve and formalize this initial idea trying to make it a tool for proving subexponential inapproximability results. 
Then, using approximation preserving sparsifiers, we derive negative results for \is{} in bounded degree graphs as well as for several fundamental problems as \ds{}, \fvs{}, etc. 
%, the \textit{superlinear sparsification}. Via this technique, we importantly strengthen Theorem~\ref{Chalermsook13thm} by showing subexponential inapproximability of \is{} within ration$n^{epsilon}$, for any $\epsilon > 0$, even in degree-bounded graphs. Then based upon either Theorem~\ref{Chalermsook13thm}, or upon an analogous result for \vc{} given in~\cite{effapprox,subexpo-fpt-inapprox-ipec13}, we devise approximability-preserving reduction that derive similar results for other well-known combinatorial problems as \ds{}, \fvs{}, \lcol{}, \cb{}, \sp{}, \textsc{min set cover}, \textsc{min hitting set}, \mmvc{}, and \mus{}. 

\section{Preliminaries}\label{prelim}

%Theorem~\ref{Chalermsook13thm} proves, in fact, that gives a first very interesting answer to an old and fundamental question whether, or not, \is{} can be approximated in subexponential time within ratios unachievable in polynomial time, since it proves that for a large range of values for the approximation ratio~$r$, indeed for $r \geqslant n^{-1/2-\delta}$, for any $\delta >0$, \is{} is not $r$-approximable in~$O^*(2^{n^{1-\delta}/r^{1+\delta}})$, unless ETH fails. 
%However, the \is{}-instance built in~\cite{2013arXiv1308.2617C} to ensure the approximability gap, is not such that this gap results in an interesting gap for other problems as, for example, \vc{}. Indeed, a 
%A careful reading of~\cite{2013arXiv1308.2617C} concludes that the negative result that could be derived for \vc{} is the impossibility of a subexponential time approximation schema. 

The idea of \textit{instance sparsification} (with respect to some parameter) has been introduced in~\cite{impa} and is very closely related to the~ETH. Informally, starting from an instance~$\phi$ of \textsc{$k$-sat}, with~$n$ variables and~$m$ clauses, the sparsification of~\cite{impa} consists of building~$2^{\epsilon n}$ (for some constant $\epsilon > 0$) ``sparse'' instances for the problem, i.e., formul{\ae} on~$n$ variables and~$\delta n$ clauses, for some $\delta > 0$, such that~$\phi$ is satisfiable if and only if one of the sparse formul{\ae} is satisfiable. Let us note that the sparsification of~\cite{impa} \textit{is {not} approximation preserving}. One of the reasons for this, is that when a clause~$C$ has all its literals contained in a clause~$C'$, a reduction rule removes~$C'$, that is safe for the satisfiability of the formula (hence, for exact computation), but not for approximation.

When handling graph problems (or problems that can be represented by means of a graph; this is, for example, the case of \textsc{min set cover}), a natural parameter upon which one can apply sparsification is the maximum degree~$\Delta$ of the input graph. So, a natural sparsification schema for such problems is to start from a graph~$G$ of order~$n$ with arbitrarily large~$\Delta$ and to produce a large number of graphs~$G_i$'s of order bounded by~$n$ and whose maximum degree~$\Delta'$ is bounded by ``something'' smaller than~$\Delta$ and such that some solution with a proved ratio for one~$G_i$ can be transformed into a solution with at least the same ratio for~$G$. Consider an instance~$G$ (of size~$n$) of an optimization problem~$\Pi$ and denote by~$\Delta$ the degree of~$G$. Let $\Pi$-$B$ denote the problem~$\Pi$ restricted to graphs with degree at most~$B$. Informally, an \textit{approximation preserving sparsification} from~$\Pi$ to~$\Pi$-$B$, maps~$G$ into a set $\{G_1,G_2,\dots,G_t\}$ of subgraphs of~$G$ and maps a solution~$S_i$ of~$G_i$ into a solution~$S$ of~$G$, this latter transformation taking polynomial time; $t \leq 2^{\epsilon n}$, for some $\epsilon > 0$, and~$G_i$'s are such that any of them has degree at most~$B_\epsilon$, for a constant~$B_\epsilon$ independent on~$n$. Furthermore, if some~$S_i$ is an $r$-approximation of $\Pi$-$B(I_i)$, then~$S$ is an $r$-approximation in~$G$.
 
In Section~\ref{sparsification} we first formalize the concept of approximation preserving sparsification and then we propose two such sparsifiers. The first sparsifier, called \textit{superlinear sparsifier}, is devised along the line informally described just above and generalizes the (linear) sparsifier introduced in~\cite{subexpo-fpt-inapprox-ipec13}. The superlinear sparsifier, in fact,  relaxes the requirement that~$B_\epsilon$ has to be constant (this was the case of the sparsifier in~\cite{subexpo-fpt-inapprox-ipec13}) and allows the sparsification tree to stop even for non-constant degrees. For simplicity, we present this sparsifier for the case of \is{} and \vc{}, but similar sparsifiers can be developed for several other problems, in particular for the APETH-equivalent problems of~\cite{subexpo-fpt-inapprox-ipec13}. One of the interesting features of this sparsifier is that it allows the transfer of negative results to problems linked to \is{}, or to \vc{},  by approximability preserving reductions building instances of size~$O(n+m)$, where~$m$ denotes the number of edges of the input graph. %for \is{} or \vc{}.
 %Theorem~\ref{Chalermsook13thm} proves, in fact, that for any $\delta >0$, and for any $r \geqslant n^{\delta-1/2}$ \is{} is not $r$-approximable in~$O^*(2^{n^{1-\delta}/r^{1+\delta}})$, unless ETH fails. Using the sparsifier of Section~\ref{subspars}, we further improve Theorem~\ref{Chalermsook13thm} by showing subexponential inapproximability of \is{} even in graphs with degree bounded by~$n^{\epsilon}$, for any $\epsilon > 0$.
The second sparsifier devised in Section~\ref{sparsification}, is called \emph{$k$-step sparsifier} and runs in polynomial time. It deals with problems whose solutions satisfy some domination property (as \is{}, \ds{}, \ids{}, and \vc{}) 
and gives quite interesting results when handling maximization problems. 

Using either superlinear or $k$-step sparsifier, together with gap-preserving reductions, we prove in Section~\ref{negres} rather strong negative subexponential inapproximability results for several fundamental problems. 
%such as \ds{}, \fvs{}, \ids{}, \textsc{min set cover}, and \textsc{min hitting set}, \is{}-$B$, etc (for, readability, definitions of all the problems discussed in the paper are given in the appendix). 
More precisely: 
\begin{itemize} 
\item via \textit{superlinear sparsifier} we show that \textit{under~ETH, and for any $\varepsilon > 0$, none of \ds{}, \textsc{min set cover} and \textsc{min hitting set}, \fvs{}, \ids{}, and \textsc{min feedback arc set} can be $(\nicefrac{7}{6}-\varepsilon)$-approximable in time~$O^*(2^{n^{1-\varepsilon}})$};
\item via \textit{$k$-step sparsifier} we show that \textit{under~ETH, for any $\varepsilon > 0$ and any $\Delta < \Delta_\varepsilon$, in $\Delta$-sparse graphs and in time $O^*(2^{O(\nicefrac{n^{1-\varepsilon}}{\Delta_{\varepsilon}^{1+\varepsilon}})})$, 
\is{}, \lcol{} and \ps{} are inapproximable within ratios $\nicefrac{\Delta}{2}-(\nicefrac{\Delta_\varepsilon}{2} - \Delta_\varepsilon^{1-\varepsilon})$, $\nicefrac{\Delta}{2}-(\nicefrac{\Delta_\varepsilon}{\ell} - \Delta_\varepsilon^{1-\varepsilon})$ and  $\Delta-(\Delta_\varepsilon - \Delta_\varepsilon^{1-\varepsilon})$, respectively};
%\end{itemize}
\item finally, using Item~\ref{Chalermsook13thm1} of Theorem~\ref{Chalermsook13thm} we show that \textit{under~{ETH}, for any $\delta>0$ and any $r \geqslant n^{\nicefrac{1}{2} - \delta}$, \mmvc{} and \ids{} are inapproximable within ratios~$\nicefrac{(c+r)}{(1+c)}$ and~$\nicefrac{1}{(1-c)}$ respectively, in less than~$O^*(2^{\nicefrac{n^{1-\delta}}{r^{1+\delta}}})$ time, in a graph of order~$nr$, with~$c$ the stability ratio of the \is{}-instance of~\cite{2013arXiv1308.2617C}.}
\end{itemize}
Our technique for proving negative results via approximation preserving sparsification (on graph problems) can be outlined as follows. Let~$\Pi$ be some problem inapproximable in time~$O^*(2^{n^{1-\epsilon}})$, for any $\epsilon > 0$,~$\Pi'$ be some problem such that~$\Pi$ reduces to~$\Pi'$ by some approximation preserving reduction~\textsf{R} that works in polynomial time and builds instances of~$\Pi'$ of size $n+m$, and~let $\mathcal{F}$ be a superlinear approximation preserving sparsifier for~$\Pi$. Then, for an instance~$G$ of~$\Pi$ we do the following:
\begin{itemize} 
\item[$-$] apply~$\mathcal{F}$ to~$G$ in order to build at most~$O^*(2^{n^{1-\epsilon}})$ $n^{\epsilon}$-sparse instances~$G_i$;
\item[$-$] transform any sparse instance~$G_i$ into an instance~$G'_i$ of~$\Pi'$;
\item[$-$] if~$\Pi$ is not approximable in time~$O^*(2^{n^{1-\epsilon}})$ within ratio~$r$ and if~\textsf{R} transforms any ratio~$r'$ for~$\Pi'$ into ratio $r = c(r')$ for some invertible function~$c$, then~$\Pi'$ is no more approximable in time~$O^*(2^{n^{1-\eta(\epsilon)}})$ within ratio~$c^{-1}(r)$.
\end{itemize}
In what follows, we use standard notation from graph theory as~$\Gamma(v)$, the set of neighbours of vertex~$v$,~$G[V']$ the subgraph of~$G$ induced by~$V'$.
%, etc., formal definotions of which are omitted. 
Given a set system~$(\mathcal{S},C)$, the frequency of the system is defined as the maximum number of subsets an element of~$C$ belongs to.
%
%we denote by~$n$ the order of a graph~$G(V,E)$ and we set $m = |E|$. We denote by~$\Delta$ the maximum degree of~$G$,  by~$\alpha(G)$ the cardinality of a maximum independent set of~$G$, and by~$\iota(G)$ the cardinality of a minimum independent dominating set of~$G$. Also, for a vertex~$v$,~$\Gamma(v)$ denotes the set of its neighbors and~$\Gamma[v]$ the set $\Gamma(v) \cup \{v\}$ and for $V' \subseteq V$, we set $\Gamma(V) = \cup_{v_i \in V}\{\Gamma(v_i)\}$. For some $V' \subseteq V$,~$G[V']$ denotes the subgraph of~$G$ induced by~$V'$; for a subgraph~$G'$ of~$G$,~$V(G')$ denotes the vertex-set of~$G'$. Given a set system~$(\mathcal{S},C)$, the frequency of the system is defined as the maximum number of subsets an element of~$C$ belongs to. Finally, for a graph problem, its instances where the maximum degree is bounded by~$B$ are called $B$-sparse instances.
%
Some of the results are given here without proofs. All missing proofs can be found in the appendix.

\section{Approximation preserving sparsifiers}\label{sparsification}

%\subsection{Some words on sparsification}\label{somewords}

We first informally describe the basic idea behind sparsification~\cite{impa} and its use for deriving lower bounds in exact computation. Assuming a reference problem~$\Pi'$ cannot be solved in~$O^*(\lambda^n)$, for some $\lambda > 1$, we are interested in showing that another problem~$\Pi$ cannot be solved in~$O^*(f(\lambda)^n)$. 
For instance, if the reference problem is \textsc{sat} and $\lambda=2$, our assumption is the Strong~ETH (SETH). 

For doing this, we use reductions from~$\Pi'$ to~$\Pi$. 
Note that one can easily derive negative results if there exists a linear reduction from~$\Pi'$ to~$\Pi$ (i.e., a reduction with linear instance-size amplification). 
But, unfortunately, linear reductions are quite rare, so that approach is limited.
Yet, reductions where~$\Pi'$ is a graph problem, amplifying the instance to a size $O(n+m)$ where~$n$ is the number of vertices and~$m$ the number of edges (or, dealing with some satisfiability problem,~$n$ is the number of variables, and~$m$ the number of clauses) are much less rare. 

A way to overcome non-linearity is to ``\textit{sparsify}'' instances of~$\Pi'$, producing, from an instance $I$, ~$\gamma(n)$ instances where the number of edges is linear to~$n$ and to prove that, for at least one of them, an optimal solution is also (or can be transformed in time at most~$O^*(\gamma(n))$ into) an optimal solution for~$I$.
 We then apply the reduction to all of these sparsified instances.

In other words, for the non-linear reductions to produce non-trivial results, we need a not too costly preprocessing step (sparsification) which makes the number of edges (resp.,~clauses) linear in the number of vertices (resp.,~variables). %, called \emph{sparsification}.

The sparsifier for \textsc{sat}, presented in~\cite{impa}, 
%(the only sparsifier known to the best of our knowledge) 
shows that \emph{for every integer $k \geqslant 3$, and every $\varepsilon > 0$ there exists a constant $C_{\varepsilon,k}$ and $2^{\varepsilon n}$ $C_{\varepsilon,k}$-sparse instances of $k$-SAT whose disjunction is equivalent to the initial instance}. But, as noticed above this idea does not work for approximation.

In Section~\ref{subspars} we extend sparsification to approximation by implementing a sparsifier for a large class of maximisation problems (whose solutions are subsets of the vertex-set of the input graph verifying some property)  that works not only for exact computation but also for approximation.

\subsection{Superlinear sparsifier}\label{subspars}

%Sparsification lies in the heart of the ETH. The sparsification lemma for \textsc{sat}~\cite{impa} reduces a CNF-formula~$\phi$ to a set of CNF-formulae~$\phi_i$, each of them having bounded occurences of variables, in such a way that solving the instances~$\phi_i$ would allow to solve~$\phi$ itself. But, as noticed in Section~\ref{prelim}, sparsification in~\cite{impa} is {not} approximation preserving due to the fact that when a clause~$C$ is contained in a clause~$C'$, a reduction rule removes~$C'$. This obviously works for the satisfiability of the formula but it does not work when considering preservation of approximation. 
%
%The \textit{approximation preserving sparsification} of~\cite{subexpo-fpt-inapprox-ipec13} was a first attempt to build a sparsifier for optimization problems and for their subexponential approximation. In what follows, we first extend the approximation preserving sparsifier of~\cite{subexpo-fpt-inapprox-ipec13} by allowing the branching to stop before the parameter becomes constant. 

Given an optimization graph problem~$\Pi$ and some parameter of the instance (this can be, for instance, the maximum, or the average degree) let $\Pi$-$B$ be the problem restricted to instances where the parameter is at most~$B$ (we use the same notations as~\cite{subexpo-fpt-inapprox-ipec13}). Then, a superlinear sparsifier can be defined as follows.
\begin{definition}\label{subapspars}
An approximation preserving superlinear sparsification from a gra\-ph prob\-lem~$\Pi$ to its bounded parameter version~$\Pi$-$B$ is a pair~$(f,g)$ of functions such that, given any function~$\phi$, sublinear in~$n$, and any instance~$G$ of~$\Pi$:
\begin{itemize}
\item $f$ maps~$G$ into a set $f(G,\phi)=(G_1,G_2,\dots,G_t)$ of instances of~$\Pi$, where $t\leq 2^{\phi(n)}$ and the orders~$n_i$ of the~$G_i$'s are all bounded by~$n$; moreover, there exists a function~$\psi$ (depending on~$\phi$) such that any~$G_i$ has parameter at most~$\psi(n)$ (for instance, if the parameter is the degree of the graph, the number of edges of~$G_i$'s is linear in~$n$, if~$\psi$ is constant, superlinear otherwise);
\item for any $i\leq t$,~$g$ maps a solution~$S_i$ of an instance $G_i \in f(G,\phi)$ into a solution~$S$ of~$G$;
\item there exists an index $i\leq t$ such that if a solution~$S_i$ is an $r$-approximation for~$G_i$, then $S=g(G,G_i,S_i)$ is an $r$-approximation for~$G$;
\item $f$ is computable in time~$O^*(2^{\phi(n)})$, and~$g$ is polynomial in~$|n|$.
\end{itemize}
\end{definition}
For simplicity, the sparsifier of Definition~\ref{subapspars} has been specified in the case of graph problems and assuming that it transfers the same ratio~$r$ from the leaves of the sparsification tree to its root. One can easily see that it can be generalized to any constant transfer function. 

It is also easy to see that the sparsifier can be easily extended to problems defined on set-systems, as \textsc{min set cover} \textsc{min hitting set}, or \textsc{max set packing}. Here, parameters can be the cardinality of the largest set, or the frequency. 
It can also be extended to fit optimum satisfiability problems, where as parameter~$B$ can be considered the maximum occurrence of a variable in the input formula. The soundness of this sparsifier relies on the following folklore lemma.
\begin{lemma}\label{lem1}
An algorithm with branching vector~$(1,\psi(n))$ where $\psi(n)=o(n)$ and $\lim\limits_{\infty} \psi=\infty$, has running time~$O^*((1+\nicefrac{1}{\psi(n)})^n)=O^*\left(2^{\nicefrac{n}{\psi(n)}}\right)$.
\end{lemma}
\begin{proof}
It is well-known that the complexity of a branching algorithm with bran\-ch\-ing vector~$(1,\psi(n))$ is~$O^*(\lambda^n)$ where~$\lambda$ is the positive solution of the equation:
\begin{displaymath}
X^{\psi(n)} - X^{\psi(n)-1} - 1 = 0
\end{displaymath}
%Let us evaluate $\lambda$. 
It holds that:
%\begin{displaymath}
$X^{\psi(n)} = \nicefrac{1}{(1-\nicefrac{1}{X})} \Longleftrightarrow \psi(n)=\nicefrac{-\log\left(1-\nicefrac{1}{X}\right)}{\log{X}}$. 
%\end{displaymath}
Set $X=1+\varepsilon$. Then~$\psi(n)$ becomes:
\begin{displaymath}
\psi(n)=\frac{-\log\left(\frac{\varepsilon}{1+\varepsilon}\right)}{\log(1+\varepsilon)}=1 - \frac{\log(\varepsilon)}{\log(1+\varepsilon)}
\end{displaymath}
Since $\lim\limits_{n \rightarrow \infty} X=1$, it holds that $\lim\limits_{n \rightarrow \infty} \varepsilon=0$. So, $\log(1+\varepsilon) \sim \varepsilon$ and thus $\psi(n) \sim -\nicefrac{\log(\varepsilon)}{\varepsilon}$ and:
$$
\log \psi(n) \sim - \log \log(\varepsilon) - \log(\varepsilon) \sim - \log(\varepsilon)
$$
%\begin{displaymath}
%f(n) \sim - \frac{\log(\varepsilon)}{\varepsilon}
%\end{displaymath}
%And,
%\begin{displaymath}
%\log f(n) \sim - \log \log(\varepsilon) - \log(\varepsilon) \sim - \log(\varepsilon)
%\end{displaymath}
Therefore, $\psi(n) \sim \nicefrac{1}{\varepsilon}$, and:
%\begin{displaymath}
%f(n) \sim \frac{1}{\varepsilon}
%\end{displaymath}
%Finally, 
%$$
$\lambda \sim 1+\nicefrac{1}{\psi(n)}$.
%$$

Note that,  since $\log(\lambda^n) = n \log (1+\nicefrac{1}{\psi(n)}) \sim \nicefrac{n}{\psi(n)}= \log (2^{\nicefrac{n}{\psi(n)}})$, $\lambda^n \sim 2^{\nicefrac{n}{\psi(n)}}$.
\end{proof}
For simplicity, we have chosen in Lemma~\ref{lem1} a very simple branching vector that fits very well many optimization problems and in particular, as Lemma~\ref{lem2} shows, \is{} and \vc{}. But the lemma works also for more general branching vectors, for instance of the form~$(\psi_1(n),\psi_2(n))$.
\begin{lemma}\label{lem2}
For any $\eta > 0$, there exists an approximation preserving $n^\eta$-spar\-s\-i\-fi\-ca\-ti\-on for \is{} and \vc{} working in time~$O^*(n^{1-\eta})$.
\end{lemma}
\begin{proof}
While the maximum degree~$\Delta$ of the surviving graph exceeds~$n^\eta$, the standard branching has vector better than~$(1,n^\eta)$ and is approximation preserving.

For \is{}, this branching consists in either including a vertex~$v$ of maximum degree to the solution and removing~$\Gamma[v]$ ($\Delta+1$ vertices are so removed), or not including~$v$ in the solution and removing it from the graph~(1 vertex removed).

For \vc{}, either include a vertex~$v$ of maximum degree in the solution and remove it from the graph~(1 vertex removed), or discard~$v$ and mandatorily include~$\Gamma(v)$ to the solution and remove~$\Gamma[v]$ ($\Delta+1$ vertices fixed).

By Lemma~\ref{lem1}, this branching takes time~$O^*(2^{n^{1-\eta}})$.
\end{proof}
One of the main characteristics of the classical notions of reducibility used for proving \textbf{NP}-completeness (i.e., Karp- or Turing-reducibility) is the superlinear amplification of the instance sizes. This fact constitutes a major drawback for using these reductions in order to transfer (in)approximability results between problems. Most of the approximation preserving reductions (see~\cite{ausiellobook} for an extensive presentation and discussion of such reductions) manage to limit this amplification in such a way that, in most cases, it remains (almost) linear. In this sense, a reduction which transforms a graph~$G$ of order~$n$ into an instance of size~$O(m)$, has very few chances to be approximation preserving (the bounded-degree requirement of the $\mathsf{L}$-reductions in~\cite{pymaxsnp} basically guarantees that~$m$ remains linear in~$n$). 

As we show in the following Theorem~\ref{nlred}, allowing the approximation preserving sparsifier to stop before the degree becomes a constant, enables us to exploit approximation preserving reductions amplifying the instance ``more than linearly'', and more precisely in $O(n+m)$. Note that, for short,  the theorem handles approximability preserving reductions from~$\Pi$ to~$\Pi'$ that transform some ratio~$r'$ for~$\Pi'$ into ratio $r = c(r') = r'$ for~$\Pi$, i.e.,~$c$ is the identity function.
\begin{theorem}\label{nlred}
Under~ETH:
\begin{enumerate} 
\item\label{mredis} if there exists an approximation preserving reduction from \is{} to a problem~$\Pi$ building instances of size $O(n+m)$, then, for any $\varepsilon>0$, and any~$r$ larger than some constant satisfying $r \leqslant n^{\nicefrac{1}{2} - \varepsilon}$,~$\Pi$ cannot be $c(r)$-approximable 
%(throughout of the paper~$c(\cdot)$ will be as informally defined in informal description of an approximability preserving reduction in Section~\ref{prelim}) 
in time $O^*(2^{\nicefrac{n^{1-2\varepsilon}}{r^{1+\varepsilon}}})$;
%, where~$1/r$ is the inapproximability bound for \is{}.
\item\label{thm3} if there exists an approximation preserving reduction from \vc{} to a problem~$\Pi$ building instances of size $O(n+m)$, then, for any $\varepsilon>0$,~$\Pi$ is not $c(\nicefrac{7}{6}-\varepsilon)$-approximable in time~$O^*(2^{n^{1-\varepsilon}})$.
\end{enumerate}
\end{theorem}
\begin{proof}
We first handle the case of reductions from \is{}. For any~$\varepsilon$, take $\eta=\varepsilon$ and apply Lemma~\ref{lem2} to obtain $n^\varepsilon$-sparse instances in time $O^*(2^{n^{1-\varepsilon}})$. Reduce all those instances to~$\Pi$; instances of size $O(n+n n^\varepsilon)=O(n^{1+\varepsilon})$ are so built.
By~\cite{2013arXiv1308.2617C}, \is{} is not $r$-approximable in~$O^*(2^{\nicefrac{n^{({1-\varepsilon})/({1+\varepsilon})}}{r^{1+\varepsilon}}})$. Thus,~$\Pi$ is not $r$-ap\-proximable in $O^*(2^{\nicefrac{n^{1-2\varepsilon}}{r^{1+\varepsilon}}})$, since $(1-2\varepsilon)(1+\varepsilon)=1-\varepsilon-2 \varepsilon^2=1-\varepsilon-o(\varepsilon)$.

\medskip

We now handle reductions from \vc{}. Beforehand let us do the following important remark. The instance of \is{} built in~\cite{2013arXiv1308.2617C} to ensure the inapproximability gap for \is{}, cannot be used to produce some gap for \vc{} that is greater than~$\nicefrac{7}{6}$, the gap of Item~\ref{vcneg} of Theorem~\ref{negipec13thm}~\cite{subexpo-fpt-inapprox-ipec13}. Indeed, using this instance, the negative result that can be derived for \vc{} is just the impossibility of a subexponential time approximation schema. So, in what follows the gap-preserving reductions from \vc{} we will use the gap~$\nicefrac{7}{6}$  of Theorem~\ref{negipec13thm}.
%In~\cite{subexpo-fpt-inapprox-ipec13}, the following stronger inapproximability bound is proved for \vc{} and will be used in the sequel.
%\begin{theorem}\label{subexpvc} \cite{subexpo-fpt-inapprox-ipec13}
%Under~{ETH}, for any  $\epsilon>0$ and any $\delta>0$, there is no
%$(7/6-\epsilon)$-ap\-p\-ro\-x\-i\-ma\-ti\-on algorithm for \vc{} running in time
%$O(2^{n^{1-\delta}})$.
%\end{theorem}
%Theorem~\ref{Chalermsook13thm} can derive a negative result for \vc{}.
%
%Baser upon Theorem~\ref{negipec13thm} and using the superlinear sparsifier, the following holds.
%\begin{theorem}\label{thm3}
%Under~ETH, if there exists an approximation preserving reduction from \vc{} to a problem~$\Pi$ with instance-size amplification $O(n+m)$, then, for any $\varepsilon>0$,~$\Pi$ is not $(\nicefrac{7}{6}-\varepsilon)$-approximable in time~$O^*(2^{n^{1-\varepsilon}})$.
%\end{theorem}
%\begin{proof}

Suppose that~$\Pi$ is $(\nicefrac{7}{6}-\varepsilon)$-approximable in time~$O^*(2^{n^{1-\varepsilon}})$ for some $\varepsilon > 0$. Apply Lemma~\ref{lem2} with $\eta=\varepsilon$ to obtain $n^\varepsilon$-sparse instances in time~$O^*(2^{n^{1-\varepsilon}})$.
Reduce all those instances to~$\Pi$;~$2^{n^{1-\varepsilon}}$ instances of size $O(n+n n^\varepsilon)=O(n^{1+\varepsilon})$ are so built.
By assumption, in time $2^{n^{1-\varepsilon}}2^{(n^{1+\varepsilon})^{1-\varepsilon}}=2^{n^{1-\varepsilon}+n^{1-\varepsilon^2}}=O(2^{n^{1-\varepsilon'}})$ (by setting, say, $\varepsilon'=2\varepsilon^2$), one can $(\nicefrac{7}{6}-\varepsilon)$-approximate all those subinstances and therefore one can $(\nicefrac{7}{6}-\varepsilon)$-approximate \vc{}, a contradiction with Item~\ref{vcneg} of Theorem~\ref{negipec13thm}.
\end{proof}

\subsection{A $k$-step sparsifier for maximization subset graph-problems}\label{polysparssec}

The superlinear sparsifier developped in Section~\ref{subspars} obviously works in superpolynomial time. In what follows, we develop, simple approximability preserving sparsifier, working in polynomial time. Here also, sparsification is done with respect to the maximum degree~$\Delta$ of the input graph~$G$. 

We deal with maximization graph problems where feasible solutions are subsets of the vertex-set verifying some specific property (in this paper we consider hereditary property); we call informally these problems ``subset problems''. Furthermore, we suppose that non-trivial feasible solutions dominate the rest of vertices of the graph.  The degree decreasing (sparsification) is done thanks to this domination characteristic of the solution. For reasons of simplicity, we describe the sparsifier for the case of \is{}, but it can be identically applied for any subset problem whose non-trivial solutions dominate the rest of the vertices of the input graph.

Consider a graph~$G$ with degree~$\Delta$ and a constant $k  < \Delta$. Then the sparsifier, builds an instance of \is{}-$\Delta - k$ running the following procedure:
\begin{quote}
\textit{for $1 \leqslant i \leqslant k$, repeatedly excavate maximal (for inclusion) independent sets~$X_i$, until the degree of the surviving graph becomes equal to~$\Delta - k$}. 
\end{quote}
Denote by~$G'(V',E')$ the instance of \is{}-$\Delta-k$, so-built. Note that, since maximal independent sets dominate the vertices of the graph where they are excavated, their removal reduces the maximum degree. Hence, at the end of the sparsification,~$G'$ has degree~$\Delta-k$. Furthermore, the sparsifier iterates~$k$ times, that is polynomial in~$n$. 
%
%Just run it for two steps, i.e., just produce a bipartite graph~$B$ dominating graph~$G[G\setminus B]$ which now has degree at most $\Delta - 2$ (for short given a graph~$G'$ subgraph of a graph~$G$, we denote by $G \setminus G'$ the subgraph of~$G$ induced by the vertices of~$G$ that do not belong to~$G'$). Then, the following holds.

Remark that non-trivial solutions of several maximization subset graph-problems verify vertex-domination property. This is the case, for instance of \lcol{}, or of \ps{}. Indeed if there exists a vertex non dominated by a vertex-set~$V'$ inducing an $\ell$-colorable subgraph, it suffices to add it in one of the color-classes. The graph $G[V' \cup \{x\}]$ always remains $\ell$-colorable. The same holds for \ps{}.
\begin{theorem}\label{main}
Let~$\mathcal{P}(\Pi,r',\Delta - k)$ be the following property: "\emph{if problem~$\Pi$ is approximable within ratio~$r'$ in time~$f(n)$ on $(\Delta-k)$-sparse graphs then, on $\Delta$-sparse graphs, it is $(r'+1)$-approximable in time~$O(f(n)+n^2)$}". Then:
\begin{enumerate}
\item\label{main1} $\mathcal{P}(\is,r',\Delta - 2)$;
\item\label{mainlcol} $\mathcal{P}(\lcol,r',\Delta - \ell)$;
\item\label{main3} $\mathcal{P}(\ps,r',\Delta - 1)$.
%The following holds for \is{}, \textsc{max $\ell$ co\-lo\-r\-a\-b\-le induced subgraph} and \ps{}:
%\begin{enumerate} 
%\item\label{mainis}
%if \is{} is $r'$-approximable in time~$f(n)$ (for some positive and increasing function~$f$) on $(\Delta-2)$-sparse graphs then, on $\Delta$-sparse graphs, it is $(r'+1)$-approximable in time~$O(f(n)+n^2)$;
%\item\label{lcol}
%If \lcol{} is $r'$-approximable in~$f(n)$ time  (for some positive and increasing function~$f$) on $(\Delta-\ell)$-sparse graphs then, on $\Delta$-sparse graphs, it is $(r'+1)$-approximable in time~$O(f(n)+n^2)$;
%\item\label{sp}
%If \ps{} is $r'$-approximable in time~$f(n)$ (for some positive and increasing function~$f$) on $(\Delta-1)$-sparse graphs then, on $\Delta$-sparse graphs, it is $(r'+1)$-approximable in time~$O(f(n))$.
\end{enumerate}
\end{theorem}
\begin{proof}
Let~$G(V,E)$ be a graph on~$n$ vertices with maximum degree~$\Delta$. 
Let~$S^*$ be a maximum independent set of~$G$. 
Run the $k$-step sparsifier for two steps and stop it (this obviously takes polynomial time). It computes two maximal independent sets~$S_1$ in~$G$, and $S_2$ in~$G[V \setminus S_1]$; $G' = G[V \setminus (S_1 \cup S_2)]$ has degree degree at most $\Delta - 2$. 
Set $B=G[S_1\cup S_2]$, the bipartite subgraph of~$G$ induced by the union of~$S_1$ and~$S_2$. 

Since~$B$ is bipartite, a maximum independent set~$S^*_B$ in~$B$ can be computed in polynomial time. 
%And, $S^*_B$ has size at least equal to the size of the part of~$S^*$ which belongs to~$B$.
If $|S^*_B| \geqslant \nicefrac{\alpha(G)}{r}$, then $S^*_B$ is an $r$-approximation \is{} in~$G$. 

Assume now that $|S^*_B| < \nicefrac{\alpha(G)}{r}$ and consider the graph $G' = G[V \setminus (S_1 \cup S_2)]$. Let $S^{*'}$ be the part of~$S^*$ contained in~$G'$. Since $|S^*_B| < \nicefrac{\alpha(G)}{r}$, and since $S^*_B$ has size at least equal to the size of the part of~$S^*$ that belongs to~$B$, $|S^{*'}| > (1-\nicefrac{1}{r})\alpha(G)$.

The graph~$G'$ has degree at most $\Delta - 2$, since if a vertex~$v$ has degree~$\Delta$, or $\Delta - 1$ in $G[V \setminus (S_1 \cup S_2)]$, then it has no neighbors in either~$S_1$, or~$S_2$ and this contradicts the maximality of at least one of them.

Run in~$G'$ the $r'$-approximation algorithm (with complexity~$f(n))$ assumed for $(\Delta-2)$-sparse graphs and denote by~$S'$ the solution returned.
Since~$S'$ is an $r'$-approximation, $|S'| \geqslant \nicefrac{|S^*|}{r'}$, so, $|S'| > ((1-\nicefrac{1}{r})\nicefrac{1}{r'})\alpha(G)$. The independent set~$S'$ is obviously a solution also for~$G$ and guarantees ratio~$\nicefrac{r\cdot r'}{r-1}$.

Finally, take the best among independent sets~$S^*_B$ and~$S'$ as solution for~$G$.

Equality of ratios~$r$ and~$\nicefrac{r\cdot r'}{r-1}$ derives $r = r'+1$.
Since ratio~$r'$ is achieved in time~$f(n)$ and the application of the sparsification step takes time~$O(n^2)$, ratio~$r$ is achieved for \is{} in~$G$ in time~$O(f(n)+n^2)$ as claimed.  %$\frac{1}{(1-\nicefrac{1}{r})}r'$ cannot be better than~$r$, otherwise the larger solution between~$S$ and~$S'$ would be an $r$-approximation in time $2^{O(n^{1-\varepsilon})}$, on $\Delta$-sparse graphs.
%
%Hence, $r < \frac{1}{(1-\nicefrac{1}{r})}r'$, deriving $r' > r - 1$. 
%\end{proof}

\medskip

%Remark that non-trivial solutions of several maximization subset graph-problems verify vertex-domination property. This is the case, for instance of \lcol{}, or of \ps{}. Indeed if there exists a vertex non dominated by a vertex-set~$V'$ inducing an $\ell$-colorable subgraph, it suffices to add it in one of the color-classes. The graph $G[V' \cup \{x\}]$ always remains $\ell$-colorable. The same holds for \ps{}.
%
%Based on the remark above, results analogous to that of Item~\ref{mainis} of Theorem~\ref{main} can be obtained for \lcol{} and \ps{}.
%\begin{proposition}\label{lcol}
%If \lcol{} is $r'$-approximable in~$f(n)$ time  (for some positive and increasing function~$f$) on $(\Delta-\ell)$-sparse graphs then, on $\Delta$-sparse graphs, it is $(r'+1)$-approximable in time~$O(f(n)+n^2)$.
%\end{proposition}
%\begin{proof}
For \lcol{}, let~$G(V,E)$ be a graph on~$n$ vertices with maximum degree~$\Delta$. 
Let~$L^*$ be an optimal solution for \lcol{} on~$G$. 
Run the the $k$-step sparsifier for \is{} for~$\ell$ steps. It iteratively excavates~$\ell$ maximal independent sets~$S_1, S_2, \ldots S_{\ell}$. 
Set $V' = S_1\cup S_2 \cup \ldots \cup S_{\ell}$, and $G'=G[V']$, the $\ell$-colorable subgraph of~$G$ induced by~$V'$. Denote by~$L^{*'}$ the part of~$L^*$ belonging to~$L^*$. 

If $|L^{*'}| \geqslant \nicefrac{L^{*}}{r}$ then, since $|V'| \geqslant |L^{*'}| \geqslant \nicefrac{L^{*}}{r}$,~$V'$ is an $r$-approximation \is{} in~$G$. 

Assume now $|L^{*'}| < \nicefrac{L^{*}}{r}$ and consider the graph $G'' = G[V \setminus V']$. Let $L^{*''}$ be the part of~$L^*$ contained in~$G''$. Since $|L^{*'}| < \nicefrac{L^{*}}{r}$, $|L^{*''}| > (1-\nicefrac{1}{r})|L^{*}|$.

The graph~$G''$ has degree at most $\Delta - \ell$ and the rest of the proof remains similar to the corresponding part of that of the first item.
%\end{proof}
%It can be easily seen that an analogous proof also holds for \ps{}. 

\medskip

For \ps{}, one just excavates only one independent set. An independent set is a planar graph. 
%So, the following holds.
The rest of the proof of the third item is the same as above.
%\begin{proposition}\label{sp}
%If \ps{} is $r'$-approximable in time~$f(n)$ (for some positive and increasing function~$f$) on $(\Delta-1)$-sparse graphs then, on $\Delta$-sparse graphs, it is $(r'+1)$-approximable in time~$O(f(n))$.
%\end{proposition}
%Theorem~\ref{isne} has also a funny impact to parameterized computation. Indeed, it is well-known that
\end{proof}

\section{Subexponential inapproximability}\label{negres}

\subsection{Via superlinear sparsification}\label{subnegres}

Combining the superlinear sparsifier of Definition~\ref{subapspars} in Section~\ref{subspars} together with approximation preserving reductions from \vc{} to several problems, the following theorem can be proved. %the problems tackled. 
%Doing this, the following theorem is proved.
\begin{theorem}\label{negresthm}
%Under~ETH, and for any $\varepsilon > 0$, none of \ds{}, \textsc{min set cover} and \textsc{min hitting set}, \fvs{} and \ids{} is $(\nicefrac{7}{6}-\varepsilon)$-approximable in time~$O^*(2^{n^{1-\varepsilon}})$.
%\begin{theorem}\label{negresthm}
Under~ETH, and for any $\varepsilon > 0$, none of \ds{}, \textsc{min set cover} and \textsc{min hitting set}, \fvs{}, \ids{}, and \textsc{min feedback arc set} is $(\nicefrac{7}{6}-\varepsilon)$-approximable in time~$O^*(2^{n^{1-\varepsilon}})$.
%\end{theorem}
\end{theorem}
\begin{proof}
%We need a reduction from Vertex Cover to Max Dominating Set which constructs instances of size $O(n+m)$ and preserves approximation.
%For Item~\ref{ena}, l
For \ds{}, let~$G(V,E)$ be an instance of \vc{} and assume~$G$ is connected. Build a graph~$G'(V',E')$ as follows. Start from a copy of~$G$ and for each edge $e=(u,v) \in E$, add two dummy vertices~$y_e$ and~$z_e$ in~$V'$ and link those vertices to~$u$ and~$v$. The graph~$G'$ so built has order $n+2m$.

A minimum dominating set in~$G'$ does not contain any dummy vertex. Indeed, if a solution~$S$ contains~$y_{(u,v)}$ or~$z_{(u,v)}$, then $S \setminus \{y_{(u,v)},z_{(u,v)}\} \cup \{u\}$ is still a dominating set of at most equal cardinality. 
%If, on the other hand~$S$ contains only, say,~$y_{(u,v)}$ the~$z_{(u,v)}$ has to be dominated this implies that~$u$ or $v$ is in $S$. So, $S \setminus \{y_{(u,v)}\}$ is still a dominating set.
%Of course, the same holds if $S$ contains only $z_{(u,v)}$.
Thus, a minimum dominating set in~$G'$ naturally maps to a subset of~$V$ which covers all the edges, hence a vertex cover of the same size. Furthermore, given an $r$-approximation of \ds{} in~$G'$, one can start by removing the potential dummy vertices as explained above, and then obtain an $r$-approximation for \vc{}. Item~\ref{thm3} of Theorem~\ref{nlred} suffices for completing the proof.

The result for \textsc{min set cover} immediately follows from a well-known approximation preserving reduction from \ds{} (function~$c$ being the identity function). Given an instance~$G(V,E)$ of \ds{}, one can construct an instance~$(\mathcal{S},C)$ of \textsc{min set cover}, where~$\mathcal{S}$ is a set-system over the ground set~$C$, by taking $\mathcal{S} = V$, $C = V$ and, for each vertex $v_i \in V$, the corresponding set $S_i \in \mathcal{S}$ contains as elements $c_j \in C$ such that vertex $v_j$ is either~$v_i$ or $v_j \in \Gamma(v_i)$.

For \textsc{min hitting set}, just observe is the problem is similar to \textsc{min set cover} where roles of~$\mathcal{S}$ and~$C$ are interchanged.

Notice that the previous reduction still works for \fvs{}. In~$G'$, every subset of vertices containing non-dummy vertex is a dominating set, iff it is a feedback vertex set\footnote{These reductions rely on the fact that, in graphs without isolated vertices, a vertex cover is both a dominating set and a feedback vertex set.}. 
%Amusingly, the previous reduction still works.
%In $G'$, every subset containing none dummy vertex is a dominating set iff it is a feedback vertex set.
%In particular, the minimum dominating sets are exactly the minimum feedback vertex sets.  
%\end{proof}

For \ids{}, tune the previous reduction by deleting all the edges in the copy of the graph~$G$. In other words, build~$G'$ from an independent set~$V$ of size $n=|V|$ where each vertex corresponds to a vertex in~$V$, and link all the vertices $u \in V$ to an independent set~$I_e$ with~2 dummy vertices for each edge $e=(u,v)$. Again, an optimal solution contains only copy vertices (no dummy vertices). Furthermore, in~$G'$, every subset containing non-dummy vertex is an independent dominating set iff it is a vertex cover in~$G$.

For \textsc{min feedback arc set}, the reduction in~\cite{karpcl} is approximation preserving with~$c$ the identity function.
The graph $G'(V',E')$ for \textsc{min feedback arc set} is built with:
\begin{eqnarray*}
V' &=& V \times \{0,1\} \\
E' &=& \{((u,0),(u,1)): u \in V\} \cup \{((u,1),(v,0)): (u,v) \in E\}
\end{eqnarray*}
In any solution, an arc $((u,1),(v,0))$ can be advantageously replaced by arc $((v,0),(v,1))$.
Indeed, a cycle containing edge~$((u,1),(v,0))$, necessarily contains also edge~$((v,0),(v,1))$ since the vertex~$(v,0)$ has out-degree~1.
Thus, removing~$((v,0),(v,1))$ destroys the same cycles (plus potentially others).
We can therefore assume that a solution is $\{((v,0),(v,1)): v \in S\}$, for some $S \subseteq V$.
Now,~$S$ is a vertex cover, and an $r$-approximation for \textsc{min feedback arc set} transforms into an $r$-approximation for \vc{}.
\end{proof}
%
%
%For \textsc{min feedback arc set}, the reduction in~\cite{karpcl} is approximation preserving.
%The graph $G'(V',E')$ for \textsc{min feedback arc set} is built with:
%\begin{eqnarray*}
%V' &=& V \times \{0,1\} \\
%E' &=& \{((u,0),(u,1)): u \in V\} \cup \{((u,1),(v,0)): (u,v) \in E\}
%\end{eqnarray*}
%In any solution, an arc $((u,1),(v,0))$ can be advantageously replaced by $((v,0),(v,1))$.
%Indeed, a cycle that contains edge~$((u,1),(v,0))$, necessarily contains also edge~$((v,0),(v,1))$ since the vertex~$(v,0)$ has out-degree~1.
%Thus, removing~$((v,0),(v,1))$ destroys the same cycles (plus potentially others).
%We can therefore assume that a solution is $\{((v,0),(v,1)): v \in S\}$, for some $S \subseteq V$.
%Now,~$S$ is a vertex cover, and an $r$-approximation for \textsc{min feedback arc set} transforms into an $r$-approximation for \vc{}.
%\end{proof}
Let us note that using the classical reduction from \vc{} to \msat{}~\cite{marathe-ravi} a similar result can be derived for \msat{}.

\subsection{Via $k$-step sparsification}

Revisit Item~\ref{chalermsook} of Theorem~\ref{Chalermsook13thm}. There,~$\Delta_\varepsilon$ is related to~$\varepsilon$ in the following way: there exists a universal constant~$C$ such that $\Delta_\varepsilon = 2^{\nicefrac{C}{\varepsilon}}$. Our purpose in this section is to strengthen this item deriving inapproximability for \is{}, \lcol{} and \ps{}, in subexponential time~$O^*(2^{O(\nicefrac{n^{1-\varepsilon}}{\Delta_{\varepsilon}^{1+\varepsilon}})})$ with a \emph{smaller} bounded degree. 
%that does not depend on~$\varepsilon$.
\begin{theorem}\label{main2}
Under~ETH, for any $\varepsilon > 0$ and any $\Delta < \Delta_\varepsilon$, in $\Delta$-sparse graphs,
\is{}, \lcol{} and \ps{} are inapproximable within ratios $\nicefrac{\Delta}{2}-(\nicefrac{\Delta_\varepsilon}{2} - \Delta_\varepsilon^{1-\varepsilon})$, $\nicefrac{\Delta}{2}-(\nicefrac{\Delta_\varepsilon}{\ell} - \Delta_\varepsilon^{1-\varepsilon})$ and  $\Delta-(\Delta_\varepsilon - \Delta_\varepsilon^{1-\varepsilon})$, respectively, in time $O^*(2^{O(\nicefrac{n^{1-\varepsilon}}{\Delta_{\varepsilon}^{1+\varepsilon}})})$.
%Under~ETH, for any $\varepsilon > 0$ and any $\Delta < \Delta_\varepsilon$: \\
%%\begin{enumerate} 
%%\item\label{main2is} 
%(I)~\is{},  is inapproximable on $\Delta$-sparse graphs within ratio $\nicefrac{\Delta}{2}-(\nicefrac{\Delta_\varepsilon}{2} - \Delta_\varepsilon^{1-\varepsilon})$ in time $O^*(2^{O(\nicefrac{n^{1-\varepsilon}}{\Delta_{\varepsilon}^{1+\varepsilon}})})$; \\
%%\item\label{main2lcol} 
%(II)~\lcol{}, is inapproximable on $\Delta$-sparse graphs within ratio $\nicefrac{\Delta}{2}-(\nicefrac{\Delta_\varepsilon}{\ell} - \Delta_\varepsilon^{1-\varepsilon})$ in time $O^*(2^{O(\nicefrac{n^{1-\varepsilon}}{\Delta_{\varepsilon}^{1+\varepsilon}})})$; \\
%%\item\label{main2ps} 
%(III)~\ps{} is inapproximable on $\Delta$-sparse graphs within ratio $\Delta-(\Delta_\varepsilon - \Delta_\varepsilon^{1-\varepsilon})$ in time $O^*(2^{O(\nicefrac{n^{1-\varepsilon}}{\Delta_{\varepsilon}^{1+\varepsilon}})})$.
%%\end{enumerate}
\end{theorem}
\begin{proof}
By Item~\ref{chalermsook} of Theorem~\ref{Chalermsook13thm}, for any $\varepsilon > 0$, \is{} on $\Delta_\varepsilon$-sparse graphs, is inapproximable within ratio~$\Delta_\varepsilon^{1-\varepsilon}$ in time $O^*(2^{\nicefrac{n^{1-\varepsilon}}{\Delta_\varepsilon^{1+\varepsilon}}})$, with $\Delta_\varepsilon = 2^{\nicefrac{C}{\varepsilon}}$ for some constant~$C$.

For any $\Delta$, run the $k$-step sparsifier on a $\Delta_\varepsilon$-sparse graph~$G$ for $\nicefrac{(\Delta_\varepsilon-\Delta)}{2}$ steps, from~$\Delta_\varepsilon$ down to~$\Delta$, in order to get a $\Delta$-sparse instance~$G'$ of \is{}. Combination of the Item~ref{main1} of Theorem~\ref{main} and of Item~\ref{chalermsook} of Theorem~\ref{Chalermsook13thm} directly derives inapproximability of \is{} in~$G$ within ratio $\Delta_\varepsilon^{1-\varepsilon}-\nicefrac{(\Delta_\varepsilon-\Delta)}{2}=\nicefrac{\Delta}{2}-(\nicefrac{\Delta_\varepsilon}{2} - \Delta_\varepsilon^{1-\varepsilon})$ in time $O^*(2^{O(\nicefrac{n^{1-\varepsilon}}{\Delta_{\varepsilon}^{1+\varepsilon}})})$. 
%\end{proof}
%Note that the inapproximability bound of Theorem~\ref{main2} cannot be derived by  Theorem~\ref{Chalermsook13thm} for $\Delta \geqslant 2^{\nicefrac{C}{\varepsilon}}(\nicefrac{1}{2}-2^{-C})$. So, Theorem~\ref{main2} extends the result of~\cite{2013arXiv1308.2617C} to degree~$\nicefrac{\Delta_\varepsilon}{2}$.

\medskip

Consider now the following simple reduction from \is{} to \lcol{}. Let~$G(V,E)$ be an instance of~\is{} of order~$n$. We keep~$G$ as the instance of \lcol{}. 
Any independent set~$S$ of~$G$ can be considered as an $\ell$-colorable graph  with empty the $\ell - 1$ of its color classes. 
Conversely, given an $\ell$-colorable graph on sets~$S_1, S_2, \ldots, S_{\ell}$, all them are independent sets and the largest among them has size more than~$\nicefrac{1}{\ell}$ times the size of the $\ell$-colorable graph. So, any ratio~$r$ for \lcol{} becomes ratio~$\ell r$ for~\is{}. 

\medskip

In the same spirit, one can devise a reduction from \is{} to \ps{}. An independent set is a planar graph per se. On the other hand since any planar graph is 4-colorable, a solution $G' = G[S]$ of \ps{} can be transformed into an independent set by coloring the vertices of~$S$ with four colors and taking the largest of them. So an approximation ratio~$r$ for \ps{} is transformed into ratio~$4r$ for \is{}.

The proofs for \lcol{} and \ps{} above of the theorem immediately derive from the remarks above.
\end{proof}
Note that the inapproximability bound for \is{} of Theorem~\ref{main2} (Item~\ref{main1}) cannot be derived by Theorem~\ref{Chalermsook13thm} for $\Delta \geqslant 2^{\nicefrac{C}{\varepsilon}}(\nicefrac{1}{2}-2^{-C})$. So, Theorem~\ref{main2} extends the result of~\cite{2013arXiv1308.2617C} to degree~$\nicefrac{\Delta_\varepsilon}{2}$.

Also, from the discussion of for \lcol{} and \ps{} in the proof of Theorem~\ref{main2}, the following corollary holds.
\begin{corollary}\label{lcolsp}
Under~ETH, and for any $\varepsilon > 0$, neither \lcol{} nor \ps{} is $r$-approximable in time~$O^*(2^{n^{\nicefrac{1-\delta}{r^{1+\delta}}}})$, where~$r$ is the approximability-gap of \is{}.
\end{corollary}

\subsection{Via Theorem~\ref{Chalermsook13thm}}\label{consequences}

Similar results as those of Corollary~\ref{lcolsp} can be obtained for several other problems linked to \is{} by approximability-preserving reductions. 

For instance, for \sp{}, take $\mathcal{S} = V$, $C = E$ and, for any set $S_i \in \mathcal{S}$, $S_i = \{c_j: e_j \text{ incident to } v_i\}$. This very classical reduction transforms any independent set of~$G$ to an equal-cardinality set-packing of~$(\mathcal{S}, C)$, and vice-versa.

For \mus{}, observe that its optimal value is an affine transformation of the optimum for \textsc{min set cover}. Since this latter problem is a generalization of \vc{} (indeed \vc{} can be seen as a \textsc{min set cover} problem where all ground elements have frequency~2), \mus{} is a generalization of \is{}. 

%Discussion just above, immediately derives the following result.
%\begin{proposition}\label{negres-is-pro}
%Under~ETH, and for any $\varepsilon > 0$, \cb{}, \lcol{}, \sp{} and \mus{} are not $1/r$-approximable in time~$O^*(2^{n^{1-\delta}/r^{1+\delta}})$, where~$1/r$ is the approximability-gap of \is{}.
%\end{proposition}

In what follows in this section, we handle inapproximability bounds for two problems that are closely linked between them, \ids{} and \mmvc{}. In fact, they are related in the same way as \is{} and \vc{}.

Let us first consider \mmvc{} and revisit the following reduction from \is{} given in~\cite{mmvc_waoa}. Given an instance~$G(V,E)$ of~\is{}, link any $v_i \in V$ to $n+1$ new vertices. The so-built graph~$H$ for \mmvc{} has size~$n^2+2n$. Then, by considering a \mmvc{}-solution for~$H$ consisting of taking the out-of-$G$ neighbors of some some independent set~$S$ of~$G$ together with $V\setminus S$ as solution for \mmvc{}, one can guarantee the following:
\begin{equation}\label{mmvc}
\begin{array}{rcl}
\sol(H) & \leqslant & n\cdot|S| + n \\
\opt(H) & \geqslant &  n\cdot\alpha(G) + n
\end{array}
\end{equation}
where~$\sol(H)$ and~$\opt(H)$ denote the sizes of an approximate and of an optimal solutions for \mmvc{}, respectively. Then, using expressions in~(\ref{mmvc}) and considering~$G$ the \is{}-instance of~\cite{2013arXiv1308.2617C}, one easily derives the following.
\begin{proposition}\label{mmvcbacicpro}
Under~{ETH}, for any $\delta>0$ and any $r \geqslant n^{\nicefrac{1}{4} - \delta}$, \mmvc{} is inapproximable within ratio~$r$ in less than~$O^*(2^{n^{\nicefrac{\nicefrac{1}{2} - \delta}{r^{1+\delta}}}})$ time.
\end{proposition}
Observe that, in the reduction above, $\Delta(H) \geqslant n \approx \sqrt{n(H)}$. So, the following corollary derives from Propostion~\ref{mmvcbacicpro}.
\begin{corollary}
Furthermore, under~{ETH}, for any $\delta>0$ and any $r \geqslant \Delta^{\nicefrac{1}{2} - \delta}$, \mmvc{} is inapproximable within ratio~$r$ in less than~$O^*(2^{\Delta^{\nicefrac{1-\delta}{r^{1+\delta}}}})$ time.
\end{corollary}
The result of Proposition~\ref{mmvcbacicpro} can be further strengthened by slightly changing the reduction of~\cite{mmvc_waoa}. 
Denote by~$c$ the stability ratio~$\nicefrac{\alpha(G)}{n}$ of~$G$. Then the following holds.
\begin{proposition}\label{mmvcproimp}
Under~{ETH}, for any $\delta>0$ and any $r \geqslant n^{\nicefrac{1}{2} - \delta}$,  in any graph of order~$nr$, \mmvc{} is inapproximable within ratio~$\nicefrac{(c+r)}{(1+c)}$ in time less than $O^*(2^{\nicefrac{n^{1-\delta}}{r^{1+\delta}}})$, where~$c$  the stability ratio of the \is{}-instance of~\cite{2013arXiv1308.2617C}.
\end{proposition}
\begin{proof}
Consider the \is{}-instance of Theorem~\ref{Chalermsook13thm} and link any of its vertices to $r+1$ new vertices where~$r$ is as in Item~\ref{Chalermsook13thm1} of Theorem~\ref{Chalermsook13thm}. The \mmvc{}-instance~$H$ has now $n(r+1)$ vertices.  Set $\rho'(H) = \nicefrac{\sol(H)}{\opt(H)}$, the inverse of the approximation ratio for \mmvc{} in~$H$. Then, using~(\ref{mmvc}), it holds that:
\begin{equation}\label{mmvcimp}
\frac{|S|}{\alpha(G)} \geqslant \rho'(H) - \frac{(1 - \rho'(H))n}{r\alpha(G)} 
\end{equation}
As one can see in the proof of Item~\ref{Chalermsook13thm1} of Theorem~\ref{Chalermsook13thm},~$\alpha(G)$ is linear in~$n$, i.e., $\alpha(G) \geqslant cn$ for some fixed (independent on~$n$) $c < 1$. So,~(\ref{mmvcimp}) becomes:
\begin{equation}\label{mmvcimp1}
\frac{1}{r} \geqslant \frac{|S|}{\alpha(G)} \geqslant \rho'(H) - \frac{(1 - \rho'(H))c}{r} \geqslant \rho'(H) - \frac{c}{r}
\end{equation}
where the first inequality above is due to the inapproximability bound~$1/r$ for \is{} in the graph of Item~\ref{Chalermsook13thm1} of Theorem~\ref{Chalermsook13thm}. Then some simple algebra derives $\rho(H) = \nicefrac{1}{\rho'(H)} \geqslant \nicefrac{c+r}{1+c}$, as claimed.
\end{proof}
Interestingly enough, although \ids{} is one of the hardest problems for polynomial approximation, only  subexponential inapproximability within ratio~$\nicefrac{7}{6}-\varepsilon$ can be proved for it, using sparsification. The following proposition gives a stronger subexponential inapproximability bound
%
%In~\cite{halmmis}, 
%%it is proved that \ids{} is inapproximable in polynomial time within ratio~$n^{1-\varepsilon}$, for any $\epsilon > 0$ unless $\textbf{P}=\textbf{NP}$. 
%%The core of the proof is a gap-reduction from \textsc{sat} to \ids{} in~\cite{Irving91}, whithout any~PCP machinery and 
%the following result is also shown regarding inapproximability of \ids{} in subexponential time: \textit{under~ETH, \ids{} is not $(\nicefrac{n}{\log^{2+\varepsilon}n})$-approximable in time $O^*\left(n^{\log^{\nicefrac{\varepsilon}{3}}n}\right)$}. Unfortunately, this bound does not tell us much about its inapproximability in (almost) tight subexponential time.
%
%In the following proposition, we give an inapproximability bound 
for \ids{} using the fact that an independent dominating set in some graph~$G$  is the complement of a minimal vertex cover of~$G$.
\begin{proposition}\label{proids}
Under~{ETH}, for any $\delta>0$ and any $r \geqslant n^{\nicefrac{1}{2} - \delta}$, in any graph of order~$nr$, \ids{} is inapproximable within ratio~$\nicefrac{1}{(1-c)}$ in time less than $O^*(2^{n^{\nicefrac{1-\delta}{r^{1+\delta}}}})$, where~$c$ is the stability ratio~$\nicefrac{\alpha(G)}{n}$ of the \is{}-instance of~\cite{2013arXiv1308.2617C}.
\end{proposition}
\begin{proof}
Consider again the graph~$G$ built in Item~\ref{Chalermsook13thm1} of Theorem~\ref{Chalermsook13thm} and the reduction of Proposition~\ref{mmvcproimp} to \mmvc{}. Denote by~$c$ the stability ratio of~$G$, i.e., $c = \nicefrac{\alpha(G)}{n}$, and recall that~$c$ is a fixed constant~\cite{2013arXiv1308.2617C}. Then:
\begin{equation}\label{idseq0}
\iota(H) = \alpha(G) + (n-\alpha(G))(r+1) = (1-c)n(r+1)
\end{equation}
 Denote by~$\iota'(H)$, the independent dominating set associated with the approximate minimal vertex cover of~$H$, i.e., $\iota'(H) = n(r+1) - \sol(H)$ and by~$b$ the inverse of the inapproximability bound for \mmvc{} ($b < 1$). Then, using~(\ref{idseq0}), we get:
%\begin{eqnarray}\label{idseq}
%& & b \geqslant \frac{\sol(H)}{\opt(H)} \geqslant \frac{n(r+1) - \iota'(H)}{n(r+1) - \iota(H)} \nonumber \\
%& \Longrightarrow& \frac{\iota'(H)}{\iota(H)} \geqslant \frac{n(r+1)(1-b)}{\iota(H)} + b \geqslant \frac{1-b}{1-c} + b \approx \frac{1}{1-c}
%\end{eqnarray}
\begin{eqnarray*}
b \geqslant \frac{\sol(H)}{\opt(H)} \geqslant \frac{n(r+1) - \iota'(H)}{n(r+1) - \iota(H)} &\Longrightarrow& \frac{\iota'(H)}{\iota(H)} \geqslant \frac{n(r+1)(1-b)}{\iota(H)} + b \\
&\geqslant& \frac{1-b}{1-c} + b \sim \frac{1}{1-c}
\end{eqnarray*}
where the last approximation for~$b$ is due to the fact that $b = o(1)$.
\end{proof}
%Let us now conclude this section by a conditional result about \ids{} that should be very likely true. For this we devise a reduction from \textsc{multicolored independent set} to \ids{}. Given a graph whose vertices are partitioned in~$k$ independent sets, \textsc{multicolored clique} asks for a clique of size~$k$. Via the well-known equivalence of \is{} and \textsc{max clique} under complementation of~$G$, \textsc{multicolored independent set} can be seen as a \textsc{multicolored clique} problem in~$\bar{G}$ and, as far as exact complexity in~$n$ is concerned, the two problems are solved with the same complexity. Note that \textsc{multicolored clique} is a very frequently used problem for proving hardness results in parameterized complexity.
%\begin{proposition}
%Either \ids{}t is $r$-approximable in time $O^*\left(2^{\frac{f(n)}{r}}\right)$, for some function~$f$, or, under ETH, \textsc{multicolored clique} is not solvable in time~$O^*(2^{f(n))}$.
%\end{proposition}
%\begin{proof}
%Let $G=(V_1 \cup \ldots \cup V_k,E)$ be any input graph for \textsc{multicolored independent set} with $G[V_i]=K_{|V_i|}$ for each $i \in \{1,\ldots,k\}$. Build the graph $G'=(V_1 \cup \ldots \cup V_k \cup I_1 \cup \ldots \cup I_k,E \cup E')$ for \ids{}, where~$I_i$ is an independent set of size~$t$ for each $i \in \{1,\ldots,k\}$, and $E'=\{(u,v) : \exists i,u \in I_i \land v \in V_i\}$.
%
%Again, we distinguish the \emph{copy} vertices in $V_i$s from the \emph{dummy} vertices in $I_i$s. 
%\end{proof}

\section{More about sparsifiers}\label{negexres}

Revisit the informal description of sparsification in Section~\ref{subspars}. The sparsifier designed in~\cite{impa} may yield very weak lower bounds, in the sense that~$f(\lambda)$ may be very close to~1. 
%We show in what follows that, via  the superlinear sparsifier, better lower bounds could be achieved.
%Furthermore, as mentioned in Section~\ref{prelim}, this specific sparsifier does not preserve approximation ratios. Hence, it can only give results for exact computations.
%
%Then, we use an exact branching as a sparsifier to get more significant relative impossibility results.
%This approximation preserving sparsifier has been introduced in \cite{subexpo-fpt-inapprox-ipec13}.
%
%\subsection{Using the sparsifier of SAT}
%
%The sparsifier for \textsc{sat}~\cite{impa} shows that \textit{for every integer $k \geqslant 3$, and every $\varepsilon > 0$ there exists a constant $C_k^\varepsilon$ and $2^{\varepsilon n}$ $C_k^\varepsilon$-sparse instances of $k$-SAT whose disjunction is equivalent to the initial instance}.
%
Suppose that there exists a polynomial time reduction~$\mathsf{R}$ from \textsc{$k$-sat} to a problem~$\Pi$, and two integers~$\alpha$ and~$\beta$ such that, for an instance~$\phi$ of \textsc{$k$-sat} with~$n$ variables and~$m$ clauses, $\mathsf{R}(\phi)$ is of size $\alpha n+\beta m$. 
To solve an instance of \textsc{$k$-sat} on~$\phi$, one can sparsify it, reduce all the~$2^{\varepsilon n}$ sparsified formul{\ae}, and solve each instance of~$\Pi$ built by application of~$\mathsf{R}$ to any sparse instance produced from~$\phi$. 
This takes time $O^*((2^{\varepsilon} \lambda^{\alpha+\beta C_k^\varepsilon})^n)$. 
Assuming~ETH, let~$\lambda_k$ be the smallest real number such that 
$k$-SAT is solvable in~$O^*(\lambda_k^n)$. 
Then, $2^{\varepsilon} \lambda^{\alpha+\beta C_k^\varepsilon} \geqslant \lambda_k$.
Adjusting~$\varepsilon$ to get the best possible lower bound for~$\lambda$, one gets $\lambda-1 < 10^{-10}$, for plausible values of~$\alpha$ and~$\beta$. So, one only shows that~$\Pi$ is not solvable in, say, $O^*((1+10^{-10})^n)$.

We show that the superlinear sparsifier of Section~\ref{subspars} may be used to produce stronger lower bounds than those get by the sparsifier of~\cite{impa}.
%
%\subsection{Using the branching-based approximation preserving sparsifier}
%
In order to do that, we will use the central problem of the paper, the \is{} problem.
Assume~$\mathcal H_{\mathrm{IS}}(\lambda)$ is the hypothesis that \is{} is not solvable in time $O^*(\lambda^n)$, and $g:(1,2) \rightarrow \mathbb{N}$ maps any real value $x$ in $(1,2)$ to the smallest integer $p$ such that the positive root $X^{p+1}-X^p-1=0$ is smaller than $x$.
The superlinear sparsifier can be used to show the following.
\begin{proposition}\label{his}
Let~$\Pi$ be problem such that there exists a polynomial time reduction~$\mathsf{R}$ from \is{} to~$\Pi$ and two positive numbers~$\alpha$ and~$\beta$ satisfying, for all instances~$G$ of \is{}, $|\mathsf{R}(G(V,E))| \leqslant \alpha |V| + \beta |E| = \alpha n + \beta m$.
Under $\mathcal H_{\mathrm{IS}}(\lambda)$, if $\Pi$ is solvable in $O^*(\mu^n)$, then $\mu > \lambda^{\nicefrac{1}{\alpha+\lfloor \nicefrac{g(\lambda)}{2} \rfloor \beta}}$ 
\end{proposition}
\begin{proof} 
Use the superlinear sparsifier with the threshold $\Delta = g(\lambda)$, that is, stop the branching when the degree of the graph becomes strictly less than~$g(\lambda)$.
The branching factor is the positive root of $X^{g(\lambda)+1}-X^{g(\lambda)}-1=0$ which, by construction, is smaller than~$\lambda$. 
At a leaf of the branching tree, if the number of vertices is $n-k$, then the number of edges in the remaining graph is at most $\lfloor \nicefrac{g(\lambda)}{2} \rfloor (n-k)$.

Thus, by performing the reduction~$\mathsf{R}$ on the instances at each leaf of the branching tree, and then solving the obtained instances of~$\Pi$, one gets an algorithm solving \is{} in time $O^*(\lambda^k\mu^{(\alpha+\lfloor \nicefrac{g(\lambda)}{2} \rfloor \beta)(n-k)})$.
So,  $\mu >  \lambda^{\nicefrac{1}{\alpha+\lfloor \nicefrac{g(\lambda)}{2} \rfloor \beta}}$, otherwise $\lambda^k\mu^{(\alpha+\lfloor \nicefrac{g(\lambda)}{2} \rfloor \beta)(n-k)} \leqslant \lambda^n$.
\end{proof}
Since the superlinear sparsifier is approximation preserving, if reduction~\textsf{R} from \is{} to~$\Pi$ preserves approximation, one can obtain relative exponential time lower bounds even for approximation issues.
The following proposition provides a lower bound to the best currently known complexity (function of the number of clauses) of \m3sat{}, under~$\mathcal H_{\mathrm{IS}}$.
Note that the best known running time for \m3sat{} is~$O^*(1.324^m)$. 
\begin{proposition}\label{hissat}
Under $\mathcal H_{\mathrm{IS}}(\lambda)$, \m3sat{} is not solvable in~$O^*(\lambda^{(\nicefrac{1}{1+\lfloor \nicefrac{g(\lambda)}{2} \rfloor})n})$.
\end{proposition} 
\begin{proof}
We recall the reduction in~\cite{pymaxsnp}. 
An instance $(G(V,E),k)$ of the decision version of \is{} is transformed into an instance of the decision version of \m3sat{} in the following way: each vertex $v_i \in V$ encodes a variable~$X_i$ and for each edge $(v_i,v_j) \in E$ we add a clause $\neg X_i \lor \neg X_j$.
Finally, we add the 1-clause~$X_i$ for all $v_i \in V$. 
In the so built instance of \m3sat{} we wish to satisfy at least $k+m$ clauses. 
%The two instances are equivalent since if you 
%instantiate the variables of the \maxsat{} in a way that it does not 
%correspond to an independent set, i.e., if one clause $C=\neg X_i \lor \neg X_j$ is 
%not satisfied, then you can switch $X_i$ to false. Having done that, you 
% only lose the 1-clause $X_i$ but you win at least the clause $C$, so your 
%solution is at least as good. Repeating this process, you end up with an 
%optimal solution which corresponds immediately to an independent set.
This reduction builts $n+m$ clauses, so $\alpha=\beta=1$. 
Hence, under $\mathcal H_{\mathrm{IS}}$, and according to Proposition~\ref{his}, one cannot solve \m3sat{} in time~$O^*(\mu^n)$ when $\mu = \lambda^{\nicefrac{1}{1+\lfloor \nicefrac{g(\lambda)}{2} \rfloor}}$.
\end{proof}
Suppose that~$\Pi$ is a problem (like \m3sat{} when considering its complexity in terms of~$m$) with a reduction from \is{} in $n+m$ ($\alpha=\beta=1$), and~$\Pi$ is solvable in $O^*(\mu^n)$. 
%Taking, for instance, $\lambda = (1.1, 1.18, 1.21)$, the corresponding infeasible values for~$\mu$ are (1.0073, 1.027, 1.038).
Then, the following table gives some values of~$\mu$ as function of~$\lambda$.
\begin{center}
\begin{tabular}{cc}
\hline
$\lambda$ & Infeasible value for $\mu$ \\
\hline
1.1 & 1.0073 \\
1.18 & 1.027 \\
1.21 & 1.038 \\
\hline
\end{tabular}
\end{center}

%Finally, let us note that the $k$-step sparsifier of Section~\ref{polysparssec} has also some potentially interesting consequences when handling parameterized issues that deserve further study of it (Appendix~\ref{paramissues}).

We conclude the paper by pointing out that the $k$-step sparsifier of Section~\ref{polysparssec} has also some interesting consequences when handling parameterized issues. \is{} can be solved in time $O^*((\Delta+1)^\alpha)$ with a standard branching algorithm~\cite{niedermeier06} (here $\alpha=\alpha(G)$ is the size of a maximum independent set, or equivalently the natural parameter for \is{}). The excavation performed by the $k$-step sparsifier can be used to obtain an algorithm running in time $O^*(2^{(\Delta-2)\alpha})$.
Indeed, one can excavate consecutively $\Delta-2$ maximal independent sets~$S_1$ to~$S_{\Delta-2}$, where each~$S_i$ is a maximal independent set in $G[V \setminus \bigcup_{k=1 \ldots i-1} S_k]$.
By hypothesis, for all~$i$, $|S_i| \leqslant \alpha$, so an exhaustive search on $\bigcup_{k=1 \ldots \Delta-2}S_i$ takes time $O^*(2^{(\Delta-2)\alpha})$.
Graph $G[V \setminus \bigcup_{k=1 \ldots \Delta-2} S_k]$ is a graph with degree~2, hence it takes polynomial time to complete a solution by finding a maximum independent set on this part of the graph.
This algorithm improves the branching algorithm for $\Delta \leqslant 4$.
as the following table shows.
\begin{center}
\begin{tabular}{ccc}
\hline
$\Delta$ & ~~~~Exhaustive branching~~~~ & Sparsification \\
\hline
3 & $4^\alpha$ & $2^\alpha$ \\
4 & $5^\alpha$ & $4^\alpha$ \\
%5 & $6^\alpha$ & $8^\alpha$ \\
\hline
\end{tabular}
\end{center}

%\bibliographystyle{splncs}
%\bibliography{biblio}

\begin{thebibliography}{10}

\bibitem{jo}
Johnson, D.S.:
\newblock Approximation algorithms for combinatorial problems.
\newblock J.~Comput. System Sci. \textbf{9} (1974)  256--278

\bibitem{motwaj}
Arora, S., Lund, C., Motwani, R., Sudan, M., Szegedy, M.:
\newblock Proof verification and intractability of approximation problems.
\newblock J.~Assoc. Comput. Mach. \textbf{45} (1998)  501--555

\bibitem{zucker}
Zuckerman, D.:
\newblock Linear degree extractors and the inapproximability of max clique and
  chromatic number.
\newblock In: Proc. STOC'06. (2006)  681--690

\bibitem{dofemcciwpec}
Downey, R.G., Fellows, M.R., {McCartin}, C.:
\newblock Parameterized approximation problems.
\newblock In Bodlaender, H.L., Langston, M.A., eds.: Proc. International
  Workshop on Parameterized and Exact Computation, IWPEC'06. Volume 4169 of
  Lecture Notes in Computer Science., Springer-Verlag (2006)  121--129

\bibitem{caihuiwpec}
Cai, L., Huang, X.:
\newblock Fixed-parameter approximation: conceptual framework and
  approximability results.
\newblock In Bodlaender, H.L., Langston, M.A., eds.: Proc. International
  Workshop on Parameterized and Exact Computation, IWPEC'06. Volume 4169 of
  Lecture Notes in Computer Science., Springer-Verlag (2006)  96--108

\bibitem{chgrogruiwpec}
Chen, Y., Grohe, M., Gr\"{u}ber, M.:
\newblock On parameterized approximability.
\newblock In Bodlaender, H.L., Langston, M.A., eds.: Proc. International
  Workshop on Parameterized and Exact Computation, IWPEC'06. Volume 4169 of
  Lecture Notes in Computer Science., Springer-Verlag (2006)  109--120

\bibitem{effapproxcah}
Bourgeois, N., {{Escoffier}}, B., {Paschos}, V.T.:
\newblock Efficient approximation by ``low-complexity'' exponential algorithms.
\newblock Cahier du LAMSADE 271, LAMSADE, Universit\'{e} Paris-Dauphine (2007)
  Available at
  \path+http://www.lamsade.dauphine.fr/cahiers/PDF/cahierLamsade271.pdf+.

\bibitem{effapprox}
Bourgeois, N., {{Escoffier}}, B., {Paschos}, V.T.:
\newblock Approximation of \textsc{max independent set}, \textsc{min vertex
  cover} and related problems by moderately exponential algorithms.
\newblock Discrete Appl. Math. \textbf{159} (2011)  1954--1970

\bibitem{CyganP10}
Cygan, M., Pilipczuk, M.:
\newblock Exact and approximate bandwidth.
\newblock Theoret. Comput. Sci. \textbf{411} (2010)  3701--3713

\bibitem{FurerGK09}
F{\"u}rer, M., Gaspers, S., Kasiviswanathan, S.P.:
\newblock An exponential time 2-approximation algorithm for bandwidth.
\newblock In: Proc. International Workshop on Parameterized and Exact
  Computation, IWPEC'09. Volume 5917 of Lecture Notes in Computer Science.,
  Springer (2009)  173--184

\bibitem{subexpo-fpt-inapprox-ipec13}
Bonnet, E., Escoffier, B., Kim, E., Paschos, V.T.:
\newblock On subexponential and fpt-time inapproximability.
\newblock In Gutin, G., Szeider, S., eds.: Proc. International Workshop on
  Parameterized and Exact Computation, IPEC'13. Volume 8246 of Lecture Notes in
  Computer Science., Springer-Verlag (2013)  54--65

\bibitem{impa}
Impagliazzo, R., Paturi, R., Zane, F.:
\newblock Which problems have strongly exponential complexity?
\newblock J.~Comput. System Sci. \textbf{63} (2001)  512--530

\bibitem{2013arXiv1308.2617C}
{Chalermsook}, P., {Laekhanukit}, B., {Nanongkai}, D.:
\newblock Independent set, induced matching, and pricing: connections and tight
  (subexponential time) approximation hardnesses.
\newblock CoRR \textbf{abs/1308.2617, abs/1308.2617} (2013)

\bibitem{DBLP:conf/focs/MoshkovitzR08}
Moshkovitz, D., Raz, R.:
\newblock Two query {PCP} with sub-constant error.
\newblock In: Proc. FOCS'08. (2008)  314--323

\bibitem{ausiellobook}
Ausiello, G., Crescenzi, P., Gambosi, G., Kann, V., Marchetti-Spaccamela, A.,
  Protasi, M.:
\newblock Complexity and approximation. Combinatorial optimization problems and
  their approximability properties.
\newblock Springer-Verlag, Berlin (1999)

\bibitem{pymaxsnp}
Papadimitriou, C.H., Yannakakis, M.:
\newblock Optimization, approximation and complexity classes.
\newblock J.~Comput. System Sci. \textbf{43} (1991)  425--440

\bibitem{karpcl}
Karp, R.M.:
\newblock Reducibility among combinatorial problems.
\newblock In Miller, R.E., Thatcher, J.W., eds.: Complexity of computer
  computations.
\newblock Plenum Press, New York (1972)  85--103

\bibitem{marathe-ravi}
Marathe, M., Ravi, S.:
\newblock On approximation algorithms for the minimum satisfiability problem.
\newblock Inform. Process. Lett. \textbf{58} (1996)  23--29

\bibitem{mmvc_waoa}
Boria, N., {Della Croce}, F., Paschos, V.:
\newblock On the \textsc{max min vertex cover} problem.
\newblock In: Proc. Workshop on Approximation and Online Algorithms, WAOA'13.
  Lecture Notes in Computer Science, Springer-Verlag (2013)

\bibitem{niedermeier06}
Niedermeier, R.:
\newblock Invitation to fixed-parameter algorithms.
\newblock Oxford Lecture Series in Mathematics and its Applications. Oxford
  University Press, Oxford (2006)

\end{thebibliography}

\newpage

\appendix

\section{Definition of the problems handled in the paper}

\begin{itemize}
\item \is{}. Given a graph~$G(V,E)$, determine a maximum cardinality set $V' \subseteq V$, such that any two vertices of~$V'$ are not adjacent in~$G$.
\item \vc{}. Given a graph~$G(V,E)$, determine a minimum cardinality set $V' \subseteq V$, such that any edge in~$E$ has at least one of its endpoints in~$V'$.
\item \ds{}.  Given a graph~$G(V,E)$, determine a minimum cardinality set $V' \subseteq V$ such that every vertex $v \in V \setminus V'$ is neighbor of some vertex in~$V'$.
\item \ids{}. Given a graph~$G(V,E)$, determine a minimum cardinality set $V' \subseteq V$ that is simultaneously an independent and a dominating set.
\item \fvs{}. Given a graph~$G(V,E)$, determine a minimum cardinality set $V' \subseteq V$, such that any cycle of~$G$ has at least one vertex in~$V'$.
\item \cb{}. Given a graph~$G(V,E)$, determine a maximum cardinality set $V' \subseteq V$ that induces a complete bipartite graph.
\item \lcol{}. Given a graph~$G(V,E)$ and some fixed constant~$\ell$, determine a maximum cardinality set $V' \subseteq V$ that induces an $\ell$-colorable graph.
\item \textsc{max planar induced subgraph}. Given a graph~$G(V,E)$, determine a maximum cardinality set $V' \subseteq V$ that induces a planar graph.
\item \textsc{min set cover}. Given a system~$\mathcal{S}$ of subsets of a ground set~$C$, determine a minimum cardinality subsystem~$\mathcal{S}'$ that covers~$C$.
\item \textsc{min hitting set}. Given a system~$\mathcal{S}$ of subsets of a ground set~$C$, determine a minimum cardinality subset $C' \subseteq C$ that hits all the sets of~$\mathcal{S}'$.
\item \sp{}. Given a system~$\mathcal{S}$ of subsets of a ground set~$C$, determine a maximum cardinality subsystem~$\mathcal{S}'$ of pairwise disjoint sets.
\item \mmvc{}. Given a graph~$G(V,E)$, determine a maximum cardinality set $V' \subseteq V$, that is a minimal (for exclusion) vertex cover of~$G$.
\item \mus{}. Given a system~$\mathcal{S}$ of subsets of a ground set~$C$, determine a maximum cardinality subsystem~$\mathcal{S}'$ such that $\mathcal{S} \setminus \mathcal{S}'$ covers~$C$.
%\item \msat{}. Given a first order formula~$\phi$ on~$n$ variables and~$m$ clauses, we look for an assignment of truth-values to the variables of~$\phi$ that satisfies a minimum number of clauses.
\item \textsc{min feedback arc set}. Given a directed graph~$G(V,E)$, determine a minimum cardinality set $E' \subseteq E$, such that any cycle of~$G$ has at least one edge in~$E'$.
\end{itemize}

\end{document}